\documentclass[11pt]{article}
\usepackage{amsmath,amstext,amsthm,amssymb,epsf}
\newtheorem{theorem}{\sc Theorem}

\newtheorem{coro}{\sc Corollary}
\newtheorem{nota}{\sc Notation}
\newtheorem{defin}{\sc Definition}
\newtheorem{rem}{\sc Remark}
\newtheorem{cla}{\sc Claim}
\newtheorem{ex}{\sc Example}

\newenvironment{claim}{\begin{cla}}{\end{cla}}
\newenvironment{corollary}{\begin{coro}}{\end{coro}}
\newenvironment{definition}{\begin{defin}}{\end{defin}}

\title{On the Average-Case Complexity of Shellsort}
\author{Paul Vit\'{a}nyi\thanks{Address: CWI, Science Park 123, 
1098 XG Amsterdam, The Netherlands.
Email: {\tt paulv@cwi.nl}}\\
CWI and University of Amsterdam}
\date{}
 
\begin{document}
\maketitle
 
\begin{abstract}
We prove a lower bound expressed in the increment sequence on the average-case 
complexity of the number of inversions 
of Shellsort. This lower bound is sharp in every case where it 
could be checked. A special case of this lower bound yields the general
Jiang-Li-Vit\'anyi lower bound.
We obtain new results e.g. determining the average-case 
complexity precisely
in the Yao-Janson-Knuth 3-pass case.
\end{abstract}

\section{Introduction}
The question of a tight general lower bound 
or upper bound on the average-case complexity of Shellsort
(due to D.L. Shell \cite{Sh59})
has been open for more than five decades \cite{Kn73}.
We use ``average'' throughout in the sense of ``arithmetic mean of a
uniform distribution,''
and the average-case complexity is the average number of inversions.
We present an average-case lower bound on the number
of inversions for a $p$-pass Shellsort with increments
$h_1, h_2, \ldots, h_p$ for every number of passes and increment sequences.

Shellsort sorts {\em in situ} a list of $n$ keys in
$p$ passes using a sequence of increments
$h_1 , \ldots , h_p$ with $n > h_1 > \cdots > h_p$. In the $k$th pass the
main list is divided in $h_k$ separate sublists
of length about $n/h_k$: if $n=\lfloor n/h_k \rfloor +r$ then the initial $r$
sublists have length $\lfloor n/h_k \rfloor +1$ and the remaining $h_k-r$
sublists have length $\lfloor n/h_k \rfloor$. 
The $h$th sublist consists of the keys in positions
$j \bmod h_k=h$  of the main list $j \in [1,n] := \{1,\ldots , n\}$.
Every sublist is sorted using a straightforward insertion sort.
The efficiency of the method is governed
by the number of passes $p$ and
the selected increment sequence $h_1 , \ldots , h_p$
satisfying $h_p =1$ to ensure sortedness of the final list.
Shellsort can be implemented using little code and does not
use a call stack and therefore some implementations
of the qsort function of the 
C standard library targeted at embedded systems use it
instead of quicksort, it is used in the uClibc library, 
Linux kernel, and bzip2 compressor \cite{Wi16}.
A complete description of the Shellsort algorithm, 
together with animation and examples, is provided in the last
reference.

\subsection{Previous Work}
Let $\log$ denote the binary logarithm $\log_2$.
All results below concern a permutation of $n$ keys (items) 
to be sorted. For the {\em worst-case complexity} of the number of 
inversions the following is known.
The original $\log n$-pass increment sequence 
$\lfloor n/2 \rfloor , \lfloor n/4 \rfloor, \ldots , 1$ 
of Shell \cite{Sh59} uses in the worst case 
$\Theta (n^2)$ inversions,
but Papernov and Stasevich \cite{PS65} showed that 
another related increment sequence uses a worst-case 
number of inversions equal to
$\Theta(n^{3/2})$. 
Pratt \cite{Pr72} extended the argument to a class of all nearly
geometric increment sequences and proved that there are
permutations of $n$ keys to be sorted that require $\Theta(n^{3/2})$
inversions for such an increment sequence and all
such permutations can be sorted with Shellsort using such an
increment sequence in an upper bound of 
$\Theta(n^{3/2})$ inversions. Therefore
the lower bound is equal to the upper bound. 
Incerpi and Sedgewick \cite{IS85} constructed a family of
$(8/\epsilon^2)\log n$-length increment sequences for 
which Shellsort sorts all permutations of $n$ keys in 
an upper bound of $O(n^{1+ \epsilon / \sqrt{\log n}})$ 
inversions for every $\epsilon > 0$. 
Poonen \cite{Po93} proved
a $\Omega (n^{1 + c / \sqrt{p}} )$
lower bound for any number $p$ of passes of Shellsort 
using any increment sequence for some $c > 0$ and showed
that this lower bound is tight for the Incerpi/Sedgewick 
increment sequence (and one due to B. Chazelle)
for $p = \Omega (\log n)$.
Since the lower bound drops to order $n \log n$ for 
$\log^2 n/ (\log \log n)^2$ passes and every pass takes at 
least $n$ steps this shows in fact
a $\Omega (n \log^2 n/ (\log \log n)^2)$ lower bound
on the worst-case number of inversions of Shellsort for
every increment sequence.
The currently best asymptotic method was found
by  Pratt \cite{Pr72}. It uses
all $\log^2 n$ increments of the form 
$2^i 3^j < \lfloor n/2 \rfloor$
to obtain a number of inversions of $\Theta(n\log^2 n)$ in the worst
case. 
Therefore the only possibility for Shellsort to sort in $O(n \log n)$
inversions for some number of passes and increment 
sequence is on the average.
For the {\em average-case complexity} little 
is known. 
In Pratt's \cite{Pr72} method with $\log^2 n$ increments
the average-case complexity 
is $\Theta (n \log^2 n)$. 
Knuth \cite{Kn73} shows
$\Theta ( n^{5/3})$ for the average-case of $p=2$ passes
and Yao \cite{Yao80} derives an expression for
the average case for $p=3$ that does not give a
direct bound but was used by Janson and Knuth
to derive an upper bound 
of $O(n^{23/15})$ on the average-case
complexity of 3-pass Shellsort for a particular increment sequence. 
In \cite{JLV00} Jiang, Li and Vit\'anyi
derived a general lower bound of $\Omega(pn^{1+1/p})$ on the
average-case complexity of $p$-pass
Shellsort. This lower bound shows that the only
possibility of Shellsort to run on the average
in $O(n \log n)$ inversions is for the number of passes $p$ to satisfy
$p = \Theta(\log n)$. Apart from this, no nontrivial results
were known for the average case until the results presented here.
A more detailed history can be found in \cite{Kn73}.

\subsection{Present Work}
We show a lower bound on the average number of inversions
of Shellsort expressed in the increment sequence used
(Theorem~\ref{theo.shelllb}).
The proof uses the fact that 
most permutations
of $n$ keys have high Kolmogorov complexity. Since the number
of inversions in the Shellsort process is not easily amenable 
to analysis, we analyze a simpler process. 
The lower bound on the number of unit moves of
the simpler process gives a lower
bound on the number of inversions of the original process.
We show that the largest number of unit moves 
of each key in the $k$th pass
of the simpler process is less than $h_{k-1}/h_k$ where $h_1, \ldots ,h_p$
is the increment sequence and $h_0=n$ (Claim~\ref{claim.mn}). Subsequently 
it is shown using the high Kolmogorov complexity of the permutation
that most keys in each pass have a number of unit moves close to
the maximum. This gives a lower bound on the total number of unit moves of the
simpler process (Claim~\ref{claim.lbn}) and hence a lower bound
on the number of inversions of the original process. This holds for the chosen single permutation.
Since all permutations but for
a vanishing fraction (with growing $n$) have this high Kolmogorov
complexity, the lower bound on the total number of inversions holds
for the average-case of the original Shellsort process 
(Theorem~\ref{theo.shelllb}). The lower bound is possibly tight
since it coincides with all known bounds.
For 2-pass Shellsort Knuth in \cite{Kn73} determined the average-case
complexity and the new lower bound
on the average complexity coincides with it (Corollary~\ref{cor.23}).
For 3-pass Shellsort Knuth and Janson \cite{JK97}, building on the
work of Yao \cite{Yao80}, gave an upper bound on the average-case
complexity  for a particular increment sequence 
and the new lower bound coincides with this 
upper bound (Corollary~\ref{cor.23}).
This yields the new result that the average-case 
complexity of Shellsort 
for this increment sequence is now determined. They 
\cite[Section 10]{JK97} 
conjecture an upper
bound on the average-case complexity for another increment sequence. 
The lower bound on the average-case complexity established here
for this sequence coincides with this upper bound (Corollary~\ref{cor.23}).
For the logarithmic increment sequences
of Shell \cite{Sh59}, Papernov and Stasevich \cite{PS65}, Hibbard \cite{Hi63}, 
and Pratt \cite{Pr72} (also reported in \cite{Kn73}), 
the lower bound on the average-case complexity
for the respective increment sequences is 
$\Omega(n \log n)$ (Corollary~\ref{cor.log}).
No upper bound on the average-case complexity is known for any 
of these increment sequences. 
For the square logarithmic increment sequence of Pratt \cite{Pr72} the
average-case complexity is known. 
Again, the lower bound given here coincides with it
(Corollary~\ref{cor.dlog}). 
A special case of the lower bound gives the Jiang-Li-Vit\'anyi general
lower bound (Corollary~\ref{cor.jlv}). 

\section{Preliminaries}\label{sect.prel}
We use the plain Kolmogorov complexity defined 
in \cite{Ko65} and denoted by $C$ in the text \cite{LV08}. 
It deals with finite binary strings, {\em strings} for short. 
Other finite objects
can be encoded into single strings in natural ways.
The following notions and notation may not be familiar to the
reader so we briefly discuss them.
The length of a string $x$ is denoted by $l(x)$. 
The {\em empty string} of 0
bits is denoted by $\epsilon$. Thus $l(\epsilon)=0$.
Let $x$
be a natural number or finite binary string according to the 
correspondence 
\[
( \epsilon , 0),  (0,1),  (1,2), (00,3), (01,4), 
(10,5), (11,6), \ldots .
\] 
Then $l(x)= \lfloor \log (x + 1) \rfloor$. 
The Kolmogorov complexity $C(x)$ of $x$ is the length of a shortest
string $x^*$ such that $x$ can be computed from $x^*$ by a
fixed universal Turing machine (of a special type called ``optimal'' 
to exclude undesirable such machines). In this way $C(x)$ is
a definite natural number associated with $x$ and
a lower bound on the length of a compressed version of it by any
known or as yet unknown compression algorithm. We also use the
conditional version $C(x|y)$.

A {\em pairing function} uniquely encodes two natural numbers into 
a single natural number by a primitive recursive bijection. 
One of the best-known is the computationally invertible 
Cantor pairing function defined by 
$\gamma(a,b) = \frac{1}{2} (a+b)(a+b+1)+b$.
This pairing function is inductively generalized to the 
Cantor tuple function $\gamma^n(a, \ldots , y,z) := 
\gamma(\gamma^{n-1}(a, \ldots , y),z)$. We use this to encode 
finite sequences of finite objects in a single natural number. 

Let ${\cal A}$ be a finite set of objects. 
We denote the cardinality of ${\cal A}$ by $|{\cal A}|$
(confusion with the absolute value notation is avoided by the context). 
The incompressibility method \cite[Chapter 6]{LV08} is used here as follows.
Let $n$ be a positive integer and $f$ an integer function such 
that $f(n)= \lceil \log |{\cal A}| \rceil$. 
Fix a $y$ (possibly $y \not\in {\cal A}$). 
We prove a certain property for a particular 
$x \in {\cal A}$ using (only)
its high Kolmogorov complexity $C(x|{\cal A},y) \geq f(n)-g(n)$ 
for a function $g(n) = o(f(n))$ and
$\lim_{n \rightarrow \infty} g(n) = \infty$.
How many $x\in {\cal A}$ are there such that $C(x|{\cal A},y)$ satisfies this lower bound?
Since there are
at most $\sum_{i=0}^{f(n)-g(n)-1} 2^i = 2^{f(n)-g(n)}-1$ 
binary programs of length
less than $f(n)-g(n)$ there are at least $(1-2^{-g(n)})2^{f(n)}+1$ 
objects $x \in {\cal A}$ such that $C(x|{\cal A},y) \geq f(n)-g(n)$. Since 
$\lim_{n \rightarrow \infty} (1-2^{-g(n)})2^{f(n)}=2^{f(n)}$ all 
but a vanishing fraction of the objects in ${\cal A}$ possess the property
involved with growing~$n$.  

It is customary to use ``additive constant $c$'' or
equivalently ``additive $O(1)$ term'' to mean a constant,
accounting for the length of a fixed binary program,
independent from every variable or parameter in the expression
in which it occurs.

\section{The Lower Bound}
A Shellsort computation consists essentially of a sequence of 
comparison and inversion (swapping) operations. We count the total number of
data movements (here inversions). 
The lower bound obtained below holds {\em a fortiori} for the number of 
comparisons---the algorithm must compare a pair of keys to decide whether 
or not to swap them. In practice the running time of the algorithm is
proportional to the number of inversions \cite{Kn73}.
Keys in the input permutation go by inversions to their final
destination. The sequences of inversions constitute insertion paths.
The proof is based on the following
intuition.  There are $n!$ different permutations of $n$ keys. Given the 
sorting process (the insertion paths in the right order) one can
recover the original permutation from the sorted list.
The length of a computable description of the sorting process
must be at least as
great as the Kolmogorov complexity of the starting permutation. 
The overwhelming majority of permutations have high Kolmogorov complexity.
Hence the overwhelming majority of sorting processes must have 
computable descriptions of at least a certain length.
Therefore the average sorting process has a 
computable description of that length which translates in the 
number of inversions.
The average number of inversions below is the expectation of the number
of inversions in the Shellsort sorting process when the permutations 
of $n$ keys are uniformly distributed.
 
\begin{theorem}\label{theo.shelllb}
Let for a Shellsort algorithm the sequence $h_1, \ldots , h_p$ be 
the increment sequence and $n$ be 
the number of keys in the list to be sorted.
The average number of inversions is 
$\Omega \left( n\sum_{k=1}^p h_{k-1}/h_k \right)$
where $h_0=n$. {\rm (}The proof shows this lower bound 
for all permutations of $n$ keys  
with probability going to 1 for $n \rightarrow \infty${\rm )}.
\end{theorem}

\begin{proof}
Let the list to be sorted consist of a permutation
$\sigma_0$ of the keys $1, \ldots , n$.
Let $A$ be a $p$-pass Shellsort algorithm with increments
$h_1 , \ldots , h_p$ such that
$h_k$ is the increment in the $k$th pass and $h_p=1$.
Denote the
permutation resulting from pass $k$ by $\sigma_k$. 
In each permutation the keys are ordered left-to-right. In
the final permutation $\sigma_p=12\ldots n$ the least key 1 is on
the left end and the greatest key $n$ is on the right end.

For $k=1,2, \ldots ,p$, the $k$th pass starts from $\sigma_{k-1}$ 
and this list (or permutation)
is divided into $h_k$ separate sublists or {\em $h_k$-chains} 
of length about $n/h_k$: 
if $n=\lfloor n/h_k \rfloor +r$ then the initial $r$
sublists have length $\lfloor n/h_k \rfloor +1$ and the remaining $h_k-r$
sublists have length $\lfloor n/h_k \rfloor$.
The $h$th $h_k$-chain ($1 \leq h \leq h_k$) consists of 
the keys in positions
$j \bmod h_k=h$ where $j$ is a position in the main list $j \in [1,n]$.
The insertion sort of an $h_k$-chain goes as follows. 
We start at the left end.
If the second key is less than the first key
then the second key is swapped with the first key. Otherwise
nothing happens. This creates a new $h_k$-chain.
If the third key is smaller than the first key or the second key
in the  new $h_k$-chain, then the third key is inserted in its correct position
in the $<$-order before the first key or in between the first key and
the second key. Otherwise nothing happens. We continue this way.
The $i$th key is inserted in its correct position in the $<$-order 
in the initial segment of the current $h_k$-chain consisting of the
first key through the $(i-1)$th key. All keys greater than the $i$th
key in this initial segment move one position to the right. This is possible
since the inserted key left a gap at the $i$th position of the current
$h_k$-chain. An
{\em inversion} is a swap of
key $i$ with key $j$ which changes list $\ldots ji \ldots$ to list
$\ldots ij\ldots .$ We can 
view the insertion above as the $i$th key changing 
place with the key before
it (an inversion), then changing place with the key 
before that (a second
inversion), and so on, until it ends up in its correct position.
The inversions involved are called its {\em insertion path}.  
By the time the final key is inserted in its correct position in the 
$<$-order the $h_k$-chain involved is sorted. 

All keys $i=1,2, \ldots ,n$
reside in a $h_k$-chain. 
Let $m_{i,k}$ be the number of inversions of key $i$ in its $h_k$-chain
in this sorting process. At the end of the sorting process the $h_k$-many
$h_k$-chains are merged to establish permutation $\sigma_k$ by   
putting the $j$th key of the $h$th sorted $h_k$-chain into
position $h+(j-1)h_k$ of permutation $\sigma_k$ ($1 \leq h \leq h_k$).
This process takes place for passes $k \in [1,p]$ resulting in 
the final permutation $\sigma_p=12 \ldots n$.
The sum 
\begin{equation}\label{eq.tni}
T=\sum_{i=1}^n \sum_{k=1}^p m_{i,k}
\end{equation}
is the total number of inversions that algorithm $A$ performs.

\begin{definition}\label{def.not}
\rm
Let $n$, $\sigma_0$, and the increment sequence 
$h_1, \ldots ,h_p$ with $h_p=1$ be 
as described above. 
At the start let key $i \in [1,n]$ be in position 
$p(i) \in [1,n]$ of $\sigma_0$.
For each $i$ define the 
$n_{i,k}$'s ($k \in [1,p]$) by the displacement $p(i)-i$:
\begin{itemize}
\item
If $p(i)-i > 0$ then
$\sum_{k=1}^p n_{i,k}h_k = p(i)-i$ with each $n_{i,k} \geq 0$ and 
$\sum_{k=1}^p n_{i,k}$ minimal. In this way $p(i)-i$ is represented
in a mixed radix system as $\sum_{k=1}^p n_{i,k}h_k$.
\item
If $p(i)-i < 0$ then
$\sum_{k=1}^p n_{i,k}h_k = p(i)-i$ with each $n_{i,k} \leq 0$ and 
$\sum_{k=1}^p|n_{i,k}|$ minimal.
In this way $p(i)-i$ is represented
in a mixed radix system as $\sum_{k=1}^p n_{i,k}h_k$ with
non-positive coefficients $n_{i,k}$.
\item
If $p(i)-i = 0$ then $n_{i,k}=0$ for
all $k$ ($k \in [1,p]$). 
\end{itemize}
\noindent
The sequence of integers $n_{1,1}, \ldots , n_{n,p}$ 
is the {\em minor sequence}. 
We define $N_i= \sum_{k=1}^p n_{i,k}$ for all 
$i \in [1,n]$ and $N=\sum_{i=1}^n |N_i|$. 
\end{definition}

\begin{claim}\label{lem.descr}
Given $n,A$ and the minor sequence 
we can computably reconstruct the original permutation $\sigma_0$.
\end{claim}
\begin{proof}
From the minor sequence and algorithm $A$ 
(containing the increment sequence we need) 
we can compute the displacements 
$p(1)-1, p(2)-2, \ldots , p(n)-n$
and therefore the permutation $\sigma_0$ from $12\ldots n$.
\end{proof}

\begin{claim}\label{claim.mn}
\rm
(i) $\frac{1}{2}N \leq T$.

(ii) 
With $h_0=n$ we have $|n_{i,k}| < h_{k-1}/h_k$
for all $i$ and $k$  ($i \in [1,n]$, $k \in [1,p]$).

(iii) For every $i \in [1,n]$ 
one can compute the $n_{i,k}$'s
in the order $n_{i,1}, \ldots , n_{i,p}$
from distance $p(i)-i$ with an algorithm of $O(1)$ bits.
\end{claim}
\begin{proof}
(i) By Definition~\ref{def.not} the quantity 
$\sum_{i=1}^n |p(i)-i|$ is the required sum of the distances the keys
have to travel from their positions in $\sigma_0$ to their
final positions in $\sigma_p = 12\ldots n$. 
Each $p(i)-i=\sum_{k=1}^p n_{i,k}h_k$ 
is expressed as the sum  of terms with
coefficients $n_{i,k}$ ($k \in [1,p]$) of a mixed 
radix representation with radices $h_1, \ldots ,h_p$.
Because of Definition~\ref{def.not} for every $i \in [1,n]$ we have
the sum $|N_i|=\sum_{k=1}^p |n_{i,k}|$ 
minimal for the coefficients of such a radix representation for
each distance $|p(i)-i|$.  
The number of inversions of each key $i\in[1,n]$ in the sorting 
process of Shellsort consists of $\sum_{k=1}^p m_{i,k}$ of 
the coefficients in $\sum_{k=1}^p m_{i,k}h_k$.
A {\em unit move} of key $i$ is the absolute value of 
a unit of an integer $n_{i,k}$. 
In the Shellsort sorting process keys move by inversions.
Since every inversion moves one key one position
forward and an adjacent key one position backward 
in its $h_k$-chain ($k \in [1,p]$) 
it is a pair of dependent unit moves equal to at most 
two independent unit moves. 
Hence $N=\sum_{i=1}^n|N_i|$ is smaller or equal to
$2T=2\sum_{i=1}^n\sum_{k=1}^p m_{i,k}$.

(ii) Assume by way of contradiction that
there exist $i,k$ ($i \in [1,n]$, $k \in [1,p]$)
such that $|n_{i,k}| \geq  h_{k-1}/h_k$.
Suppose $k=1$. Since $h_0=n$ we have $|n_{i,1}|h_1 \geq n$. Therefore 
$|p(i)-i| \geq n$ which is impossible. Hence $k \in [2,p]$.
Let $n_{i,k}$ be positive. Since $n_{i,k}$ is integer this implies
$n_{i,k} \geq  \lceil h_{k-1}/h_k \rceil$.
Define $n'_{i,k-1} := n_{i,k-1} +1$
and $n'_{i,k} := n_{i,k}- \lceil h_{k-1}/h_k \rceil$
while $n'_{i,h}=n_{i,h}$ otherwise.
Since $h_{k-1} > h_k$ we have $\lceil h_{k-1}/h_k \rceil > 1$
and therefore $\sum_{j=1}^p n'_{i,j} < \sum_{j=1}^p n_{i,j} = N_i$
contradicting the minimality of $N_i$. 
Let $n_{i,k}$ be negative. Since $n_{i,k}$ is integer this implies
$n_{i,k} \leq  -\lceil h_{k-1}/h_k \rceil$.
Define $n'_{i,k-1} := n_{i,k-1} -1$
and $n'_{i,k} := n_{i,k}+ \lceil h_{k-1}/h_k \rceil$
while $n'_{i,h}=n_{i,h}$ otherwise.
Since $h_{k-1} > h_k$ we have $\lceil h_{k-1}/h_k \rceil > 1$
and therefore $\sum_{j=1}^p |n'_{i,j}| < \sum_{j=1}^p |n_{i,j}| = |N_i|$
contradicting the minimality of $|N_i|$.

(iii) The representation $\sum_{k=1}^p n_{i,k}h_k$ 
of $p(i)-i$ for every $i$ ($i \in [1,n]$) is 
represented in a mixed radix system with
radices $h_1,h_2, \ldots , h_p$ and the $n_{i,k}$ of the same sign 
(or 0) for all $k \in [1,p]$. 
The question asked is whether this 
representation can be uniquely retrieved from $p(i)-i$. 
Computing the minimum of $\sum_{k=1}^p n_{i,k}$ for sequences
$n_{i,1}, \ldots , n_{i,p}$ satisfying $\sum_{k=1}^p n_{i,k}h_k=p(i)-i$ 
given $h_1,h_2, \ldots , h_p$ can be done by a
program of constant length by trying all finitely many possibilities.
\end{proof}

There are $n! \approx \sqrt{2\pi n} {n \choose e}^n$ 
permutations of $n$ keys
by Stirling's approximation. This implies 
$\log n! \approx n \log n - 1.5n$.
Choose the permutation $\sigma_0$ such that its conditional Kolmogorov
complexity (Section~\ref{sect.prel}) satisfies
\begin{equation}\label{eq.compl}
C(\sigma_0 | n,A,P) \geq n \log n - 3n,
\end{equation}
with fixed $n,A,P$, where from $n$ we determine the set ${\cal A}$ of all
permutations of $n$ keys such that $\sigma_0 \in {\cal A}$, 
the algorithm $A$ is used in 
this $p$-pass Shellsort (including the increment sequence), 
and $P$ is a constant-size algorithm to process all the information 
and to output $\sigma_0$. We use a pairing function 
to encode the conditional in a single
natural number (Section~\ref{sect.prel}). 

Denote the minor sequence $n_{1,1}, \ldots, n_{n,p}$ by $S_n$.
The description of $S_n$ comprises the displacements 
$p(1)-1, p(2)-2, \ldots , p(n)-n$
from which the minor sequence can be extracted.
A computable description of $S_n$, given 
$n,A$ and $P$, 
requires at most
\begin{equation}\label{eq.descrN}
l(descr(S_n))= (\sum_{i=1}^n \sum_{k=1}^p  \log |n_{i,k}|) +D
\end{equation}
bits. Here $D$ is the number
of bits required to be able to parse the main part of
$descr(S_n)$ into its constituent parts.
By Claim~\ref{lem.descr} we can compute permutation $\sigma_0$ 
from $descr(S_n)$, given $n,A$ and $P$.
Hence 
\begin{equation}\label{eq.NC}
l(descr(S_n)) \geq C(\sigma_0|n,A,P).
\end{equation}
From \eqref{eq.compl} and \eqref{eq.NC} it follows that
\begin{equation}\label{eq.ub}
l(descr(S_n)) \geq n \log (n/8).
\end{equation}
\begin{claim}\label{claim.lbn}
\rm
Writing $h_0=n$ we have
\[
\sum_{i=1}^n \sum_{k=1}^p |n_{i,k}| = 
\Omega \left( n \sum_{k=1}^p h_{k-1}/h_k \right).
\]
\end{claim}
\begin{proof}
By Claim~\ref{claim.mn} item (ii) 
for every $i \in [1,n]$ and every pass $k \in [1,p]$ we have
$|n_{i,k}| < h_{k-1}/h_k$. Since
$\prod_{k=1}^p h_{k-1}/h_k= h_0/h_p =n$ we have by 
\eqref{eq.descrN} and \eqref{eq.ub} that
\begin{align}\label{eq.cd}
\frac{l(descr(S_n))}{n}& = 
\frac{\sum_{k=1}^p n(\log (h_{k-1}/h_k)-a_k)}{n}+ \frac{D}{n}
\\& = \log n - \sum_{k=1}^p a_k +\frac{D}{n}
\geq \log \frac{n}{8},
\nonumber
\nonumber
\end{align}
where $a_k= \log (h_{k-1}/h_k) - 1/n \sum_{i=1}^n \log |n_{i,k}|$ 
for $k=1,2, \ldots, p$ and $a_k > 0$ by 
Claim~\ref{claim.mn} item (ii). 
\begin{definition}
\rm
The {\em self-delimiting} encoding of string $x$ is $1^{|x|}0|x|x$. 
If the length of $x$ is equal $\log n$ then its self-delimiting 
encoding has length $\log n+2 \log \log n +1$. 
\end{definition}
For each $i \in [1,n]$ 
we have $|p(i)-i| < n$ (the displacement of a key cannot be as great or greater
than the length of the list) and 
the sequence $n_{i,1}, \ldots , n_{i,p}$ 
can be extracted from the self-delimiting encoding of $p(i)-i$
using the information in $D$. 
We now show that $D/n = o(\log n)$. 
\begin{itemize}
\item
The information $D$ accounts
for the at most $(2 \log \log n +1)$-length part of the self-delimiting
encoding of $|p(i)-i|$ ($i \in [1,n]$). 
\item
We require one time 
$O(1)$ bits for a self-delimiting program to extract the 
sequences $|n_{i,1}|, \ldots ,|n_{i,p}|$ 
from the $|p(i)-i|$ ($i \in [1,n]$).
The extraction can be done for all $i \in [1,n]$ by a single 
$O(1)$-bit program by
Claim~\ref{claim.mn} item  (iii).
\item
Since all $n_{i,1}, \ldots ,n_{i,p}$  have the same sign 
for each  $i \in [1,n]$ we require $O(n)$ self-delimiting bits to 
encode them all.
\end{itemize}
To parse $descr(S_n)$ 
it therefore suffices that the quantity $D \leq 2\sum_{i=1}^n \log \log n +O(n)
=  2n \log \log n +O(n)$.
The total of the description $D$ is $o(n \log n)$ bits.

Hence up to lower order terms the last inequality of \eqref{eq.cd}
is rewritten as $\sum_{k=1}^p a_k \leq 3$. 
Since $a_k > 0$ for every $k \in [1,p]$ we have 
$0 < a_k \leq 3$.
Writing $a_k$ out and reordering this gives up to lower order terms 
\[
\log (h_{k-1}/h_k) \leq 1/n \sum_{i=1}^n \log |n_{i,k}|+3,
\]
and by exponentiation of both sides of the inequality one obtains 
\[ 
h_{k-1}/h_k = O((\prod_{i=1}^n |n_{i,k}|)^{1/n}).
\]
By the inequality of the arithmetic
and geometric means and rearranging we obtain
$\sum_{i=1}^n |n_{i,k}| = \Omega( n h_{k-1}/h_k)$ for every $1 \leq k \leq p$.
Therefore, 
$N=\sum_{k=1}^p \sum_{i=1}^n |n_{i,k}| = \Omega (n \sum_{k=1}^p h_{k-1}/h_k)$.
\end{proof}

Since $\frac{1}{2}N \leq T$ by Claim~\ref{claim.mn} item (i),
a lower bound for $\frac{1}{2}N$ is also a lower bound for $T$. Therefore 
Claim~\ref{claim.lbn} proves the statement of the theorem for the particular
 $\sigma_0$. 

By Stirling's approximation $\log n! \approx n \log (n/e)+ \frac{1}{2} \log n+O(1)
\approx n \log n - 1.44n+ \frac{1}{2} \log n +O(1)$. Therefore
$n \log n - 1.5 n \leq \log n! \leq n \log n - n$ for large $n$. 
Therefore  
by \cite[Theorem 2.2.1]{LV08} which uses a simple counting argument, 
(see also Section~\ref{sect.prel}) at least
a $(1-2^{-n})$-fraction of all permutations $\sigma$ on $n$ keys 
satisfy \eqref{eq.compl}. 
Hence the desired lower bound
holds on the average (expected over the uniform distribution) 
number of inversions.
In fact, it holds for all permutations of $1,\ldots ,n$
with probability going to 1 with growing $n$.
\end{proof}

\begin{corollary}\label{cor.23}
\rm
Set $h_0=n$.
For $p=2$ with $h_1=n^{1/3}$, and $h_2=1$ this yields
\[
T = \Omega(n (n^{1-1/3}+n^{1/3}) = \Omega(n^{5/3}),
\]
which coincides with the best number of inversions for 2-pass Shellsort
$T= \Theta(n^{5/3})$ using the same increment sequence $h_1=n^{1/3}, h_2=1$ 
as given by \cite{Kn73}.

For $p=3$ with $h_1=n^{7/15}$, $h_2=n^{1/5}$, and $h_3=1$ this yields
\[
T=\Omega(n(n^{1-7/15}+n^{7/15 - 1/5}+n^{1/5}) = \Omega(n^{1+8/15}) = 
\Omega(n^{23/15}).
\]
The upper bound of $O(n^{23/15})$ 
for 3-pass Shellsort using the same increment sequence $h_1=\Theta(n^{7/15}), 
h_2=\Theta(n^{1/5}), h_3=1$ with the additional restriction that 
$\gcd(h_1,h_2)=1$ is given in \cite{JK97}. This reference uses a complicated
probabilistic analysis based on the still more complicated combinatorial
characterization in \cite{Yao80}. 
Together with the lower bound we establish the new fact that the
average number of inversions of 3-pass Shellsort with this increment
sequence is $\Theta (n^{23/15})$.

In \cite[Section 10]{JK97} it is conjectured that with $h_1 \approx n^{1/2}$
and $h_2 \approx n^{1/4}$ ($h_3=1$) one may obtain an average-case 
number of inversions of $O(n^{3/2})$. Using the theorem above shows that
$T= \Omega (n(n^{1-1/2}+n^{1/2- 1/4}+n^{1/4}) = \Omega(n^{3/2})$.
Therefore, if the conjecture on the upper bound is true then 
3-pass Shellsort has an average-case number of inversions of 
$\Theta(n^{3/2})$ for this increment sequence.
\end{corollary}
\begin{corollary}\label{cor.log}
\rm
The increment sequence $h_1, \ldots, h_p$ with $p= \lfloor \log n \rfloor$
of Papernov and Stasevich in \cite{PS65}
is $h_1 = n/2+1, h_2 = n/2^2 +1, \dots, 
h_p=n/2^{\lfloor \log n \rfloor}+ 1$.
The worst-case number of inversions
reported by \cite{PS65} is $\Theta(n^{3/2})$. 
Since $h_{k-1}/h_k \approx 2$,
the theorem above gives
a lower bound on the average number of inversions of 
$T= \Omega (n \sum_{k=1}^{\Theta ( \log n)} \Omega (1))=
\Omega(n \log n)$.

The increment sequence of Hibbard \cite{Hi63}
with increment sequence $2^k-1$ until it passes $n$ has a worst-case
number of inversions $\Theta(n^{3/2})$. With a similar analysis
as before it gives a lower bound on the average-case of $T=\Omega( n \log n)$.
It is conjectured in \cite{We91} to lead to an average-case number of 
inversions of $O(n^{5/4})$ as reported in \cite{Kn73}. 
This conjecture is difficult to
settle empirically since for $n=100,000$ we 
have $\log n \approx n^{1/4}$.

Pratt's logarithmic increment sequence (one of his
``hypergeometric'' sequences) in \cite{Pr72} also
reported by \cite{Kn73} is 
$h_1, \ldots, h_p$ with $h_k=(3^k-1)/2$ not greater than $\lceil n \rceil$.
This increment sequence 
leads to a worst-case number of inversions of $\Theta(n^{3/2})$.
In this case $h_{k-1}/h_k \approx 3$ 
and the number of passes is $p=\log_3 n$.
The theorem above gives
a lower bound on the average number of inversions of 
$T= \Omega (n \sum_{k=1}^{\Theta (\log n)} \Omega (1))=
\Omega(n \log n)$.

The original
increment sequence used by Shell \cite{Sh59} was $\lfloor n/2\rfloor,
\lfloor n/2^2 \rfloor$, and so on for $\log n$ passes. 
Knuth \cite{Kn73} remarks that this is
undesirable when the binary representation of $n$ 
contains a long string
of zeroes and has the effect that Shellsort runs in worst-case time 
$\Theta(n^2)$. 
By the same analysis as given above for the 
Papernov-Stasevich increment sequence
the lower bound on the average number of inversions 
is $\Omega(n \log n)$. 

By \cite{JLV00} the average number of inversions of Shellsort
can be $\Theta(n \log n)$ only for
an increment sequence $h_1, \ldots, h_p$ with $p=\Theta(\log n)$.
We have shown here that the lower bound on the
average number of inversions is $\Omega(n \log n)$ for many increment sequences 
of this length.
It is an open problem whether it can be proved that for some such
increment sequence the average number of inversions is $O(n \log n)$. 
\end{corollary}
\begin{corollary}\label{cor.dlog}
\rm
For  Pratt's square logarithmic increment sequence $h_1, \ldots, h_p$ 
with $p= \Theta( (\log n)^2)$,
the average-case number of inversions 
is lower bounded by
$T= \Omega (n \sum_{k=1}^{\Theta ((\log n)^2)} \Omega (1))=
\Omega(n (\log n)^2)$. The precise average-case number (and worst-case
number) of inversions 
is $\Theta(n (\log n)^2)$ in \cite{Pr72},
and therefore the lower bound is tight. 
\end{corollary}
\begin{corollary}\label{cor.jlv}
\rm
The theorem enables us to establish asymptotic 
lower bounds for $n \rightarrow \infty$ 
keys and $p$ passes. 
For example choose for a $p$-pass Shellsort the 
increment sequence with identical ratios between increments 
$h_1=n^{1-1/p}, h_2=n^{1-2/p}, \ldots , h_p=n^{1-p/p}=1$.
With $h_0=n$ the sum becomes $\sum_{k=1}^p h_{k-1}/h_k=pn^{1/p}$. 
The lower bound for $p$ passes
becomes $T=\Omega(pn^{1+1/p})$, that is, the lower bound of \cite{JLV00}.
This lower bound is a greatest
lower bound which holds for all increment sequences of $p$ passes. 
This can be seen as follows.
For increment sequences we express the increments 
as real powers of $n$.
If the ratios between successive increments are unequal then there is
one of those ratios which is  greater than other ratios.
If $h_{k_0-1}/h_{k_0}$ is such a  maximum ratio 
($k_0 \in [1,p]$) which means that for some $\epsilon > 0$
we have $h_{k_0-1}/h_{k_0}=n^{1/p+\epsilon}>n^{1/p}$, then  
the lower bound becomes $T=\Omega(nh_{k_0-1}/h_{k_0})) \neq
\Omega(pn^{1+1/p})$. 

We give an example of an increment sequence for $p=4$ which has a lower
bound greater than $\Omega(pn^{1+1/p})=\Omega(n^{5/4})$. 
We choose the increment sequence
$h_1=n^{11/16}, h_2=n^{7/16}, h_3=n^{3/16}, h_4=1$. The lower bound becomes
$T= \Omega(n\cdot (n^{1-11/16}+n^{11/16-7/16}+n^{7/16-3/16}+n^{3/16})=
\Omega(n^{1+5/16})=\Omega(n^{21/16}) \neq \Omega(n^{20/16})$. 
\end{corollary}

\section{Conclusion}
A first nontrivial general lower bound on the average-case number of 
inversions of Shellsort using
$p$ passes was given in \cite{JLV00}.
Here we gave a lower bound on the average-case complexity for
each increment sequence separately. 
The lower bound of above reference turns out to be the greatest
lower bound which holds for all increment sequences. In fact,
the lower bound given here seems to be possibly 
tight as follows from the corollaries.
A tantalizing
prospect is to obtain a lower bound for the best increment sequences 
expressed only in the number of keys to be sorted and the
number of passes in the sorting process, and which is tighter than
the lower bound in the quoted reference.

\section*{Acknowledgment}
I thank Ronald de Wolf for comments on a preliminary draft of this paper
and the referees for constructive suggestions.


{\small
\begin{thebibliography}{99}

\bibitem{Hi63}
T.N. Hibbard, An Empirical Study of Minimal Storage Sorting, 
{\em Commun. ACM}, 6:5(1963), 206–-213. 


\bibitem{IS85}
J. Incerpi and R. Sedgewick,
Improved upper bounds on Shellsort,
{\em Journal of Computer and System Sciences},
31(1985), 210--224.

\bibitem{JK97}
S. Janson, D. Knuth,
Shellsort with three increments,
{\em Random Structures Algorithms}, 10:1-2(1997), 125--142.

\bibitem{JLV00}
T. Jiang, M. Li, P.M.B. Vit\'anyi, A lower bound on the average-case 
complexity of Shellsort, {\em J. Assoc. Comp. Mach.}, 47:5(2000), 905--911. 
\bibitem{Kn73}
D.E. Knuth, {\em The Art of Computer Programming, Vol.3:
Sorting and Searching}, Addison-Wesley, 1973 (1st Edition),
1998 (2nd Edition).

\bibitem{Ko65}
A.N. Kolmogorov,
Three approaches to the quantitative definition of information.
{\em Problems Inform. Transmission}, 1:1(1965), 1--7.

\bibitem{LV08}
M. Li and P.M.B. Vit\'anyi, 
{\em An Introduction to Kolmogorov Complexity
and its Applications},
Springer-Verlag, New York, 3rd Edition, 2008.

\bibitem{PS65}
A. Papernov and G. Stasevich,
A method for information sorting in computer memories,
{\em Problems Inform. Transmission}, 1:3(1965), 63--75.


\bibitem{Po93}
B. Poonen, The worst case in Shellsort and related algorithms,
{\em J. of Algorithms}, 15(1993), 101-124.

\bibitem{Pr72}
V.R. Pratt, {\em Shellsort and Sorting Networks},
Ph.D. Thesis, Stanford University, 1972.



\bibitem{Sh59}
D.L. Shell, A high-speed sorting procedure,
{\em Commun. ACM}, 2:7(1959), 30--32. 


\bibitem{Wi16}
Shellsort in Wikipedia in December 2016,
http://en.wikipedia.org/wiki/Shellsort, 2016. 

\bibitem{We91}
M.A. Weiss, Empirical study of the expected
running time of Shellsort, {\em Comput. J.}, 34(1991), 88--91

\bibitem{Yao80}
A.C.C. Yao, An analysis of $(h,k,1)$-Shellsort,
{\em J. of Algorithms}, 1(1980), 14--50.
\end{thebibliography}
}
\end{document}